\documentclass[a4paper,UKenglish]{lipics-v2016}

\usepackage{microtype}

\title{Confluence of an extension of Combinatory Logic by Boolean
  constants\footnote{Supported by Marie Sk{\l}odowska-Curie action ``InfTy'',
    program H2020-MSCA-IF-2015, number~704111, and by Narodowe Centrum
    Nauki grant 2012/07/N/ST6/03398.}}

\author[1]{{\L}ukasz Czajka}
\affil[1]{DIKU, University of Copenhagen, Copenhagen, Denmark\\
  \texttt{luta@di.ku.dk}}
\authorrunning{{\L}. Czajka} 

\Copyright{{\L}. Czajka}

\subjclass{F.4.2 Grammars and Other Rewriting Systems}
\keywords{combinatory logic, conditional linearization, unique normal
  form property, confluence}

\EventEditors{Dale Miller}
\EventNoEds{1}
\EventLongTitle{2nd International Conference on Formal Structures for Computation and Deduction (FSCD 2017)}
\EventShortTitle{FSCD 2017}
\EventAcronym{FSCD}
\EventYear{2017}
\EventDate{September 3--9, 2017}
\EventLocation{Oxford, UK}
\EventLogo{}
\SeriesVolume{84}
\ArticleNo{14} 

\newcommand{\Ps}{\mathsf{P}}
\newcommand{\Cs}{\mathsf{C}}
\newcommand{\Fs}{\mathsf{F}}
\newcommand{\Ts}{\mathsf{T}}
\newcommand{\Ss}{\mathsf{S}}
\newcommand{\Ks}{\mathsf{K}}
\newcommand{\Is}{\mathsf{I}}

\newcommand{\from}{\ensuremath{\leftarrow}}

\newcommand{\From}{\ensuremath{\Leftarrow}}

\newcommand{\Da}{\ensuremath{\Downarrow}}

\newcommand{\reduces}{\ensuremath{\to^*}}
\newcommand{\contr}{\ensuremath{\rightarrow}}
\newcommand{\eqv}{\ensuremath{\leftrightarrow}}
\newcommand{\erase}[1]{\ensuremath{|#1|}}

\newcommand{\Var}{\mathrm{Var}}
\newcommand{\Nbb}{\ensuremath{\mathbb{N}}}

\newcommand{\la}{\ensuremath{\langle}}
\newcommand{\ra}{\ensuremath{\rangle}}

\newcommand{\leftidx}[3]{{\;\vphantom{#2}}#1\!\!#2#3}

\newcommand{\Sc}{\ensuremath{{\mathcal S}}}

\newcommand{\CLC}{{\mbox{${\text{\rm CL-pc}^1}$}}}
\newcommand{\CLCz}{{\mbox{$\text{\rm CL-pc}$}}}
\newcommand{\CLCp}{{\mbox{$\text{\rm CL-pc}^{\mathrm{L}}$}}}
\newcommand{\sCLC}{{\ensuremath{\text{\rm CL-pc}^1}}}
\newcommand{\sCLCz}{\ensuremath{\text{\rm CL-pc}}}
\newcommand{\sCLCp}{\ensuremath{\text{\rm CL-pc}^{\mathrm{L}}}}
\newcommand{\CLCs}{{\mbox{${\text{\rm CL-pc}^s}$}}}

\newcommand{\equiverased}{\succ}

\begin{document}

\maketitle

\begin{abstract}
  We show confluence of a conditional term rewriting system~$\CLC$,
  which is an extension of Combinatory Logic by Boolean
  constants. This solves problem~15 from the~RTA list of open
  problems. The proof has been fully formalized in the~Coq proof
  assistant.
\end{abstract}

\section{Introduction}

Combinatory Logic is a term rewriting system defined by two rules:
\[
\begin{array}{rclcrcl}
  \Ks x y &\to& x &\quad\quad& \Ss x y z &\to& x z (y z)
\end{array}
\]
Using only~$\Ss$ and~$\Ks$, it is possible to encode natural numbers via
Church numerals. Any computable function may then be represented by a
term in the system. However, a conditional~$\Cs$ encoded in this way
does not have a desirable property that $\Cs t_1 t_2 t_2 = t_2$ if $t_1$
encodes neither true nor false. It is therefore interesting to
investigate extensions of Combinatory Logic incorporating a
conditional directly. Perhaps the most natural such extension
is~$\CLCz$:
\[
\begin{array}{rclcrclcrcl}
  \Ks x y &\to& x &\quad& \Cs \Ts x y &\to& x &\quad& \Cs z x x &\to& x\\
  \Ss x y z &\to& x z (y z) &\quad& \Cs \Fs x y &\to& y &&&&
\end{array}
\]
The system~$\CLCz$ is known to be not confluent~\cite{Klop1980}. One
may thus try other ways of adding a conditional and Boolean constants
to Combinatory Logic.

We show confluence of a conditional term rewriting system~$\CLC$
defined by the rules:
\[
\begin{array}{rclcrclcrcl}
  \Ks x y &\to& x &\quad& \Cs \Ts x y &\to& x &\quad& \Cs z x y &\to& x
  \quad\From\quad x = y\\
  \Ss x y z &\to& x z (y z) &\quad& \Cs \Fs x y &\to& y &&&&
\end{array}
\]
Confluence of this system\footnote{Strictly speaking, in the
  literature the systems~$\CLCz$, $\CLC$ and~$\CLCp$ also contain the
  rule $\Is x \to x$. This rule could be added to our definitions
  without significantly changing the proofs. However, this would
  increase the number of cases to consider, making the proofs less
  readable. The formalization of our results uses the definitions from
  the literature.} appears as problem~15 on the RTA list of open
problems~\cite{DershowitzJouannaudKlop1991}.

The equality in the side condition for the third rule for~$\Cs$
in~$\CLC$ refers to equality in the system~$\CLC$ itself, thus the
definition is circular. This circularity is an essential property
of~$\CLC$ which distinguishes it from~$\CLCz$.

A system related to~$\CLC$ is~$\CLCp$, which consists of all rules
of~$\CLC$ plus:
\begin{eqnarray*}
  \Cs z x y \to y &\From& x = y
\end{eqnarray*}
It is known that~$\CLCp$ is confluent~\cite{Vrijer1999}. However, the
confluence proof in~\cite{Vrijer1999} essentially depends on a
``semantic'' argument to first establish $\Ts \ne_{\sCLCp} \Fs$. We
provide a ``syntactic'' proof of confluence of both~$\CLC$
and~$\CLCp$.

The systems~$\CLC$ and~$\CLCp$ are conditional linearizations
of~$\CLCz$. The notion of conditional linearization was introduced in
the hope of providing a simpler proof of Chew's
theorem~\cite{Chew1981,ManoOgawa2001} which states that all compatible
term rewriting systems have the unique normal form (UN)
property. Compatibility imposes certain restrictions on the term
rewriting system, but it does not require termination or
left-linearity. In particular, Chew's theorem is applicable to many
term rewriting systems which are not confluent. For instance,~$\CLCz$
satisfies the conditions of Chew's theorem, but it is not
confluent. As shown in~\cite{Vrijer1999}, to prove the unique normal
form property of a term rewriting system, it suffices to prove
confluence of one of its conditional linearizations. The proof of
Chew's theorem in~\cite{ManoOgawa2001} is quite complicated and uses a
related but different approach, relying on left-right separated
conditional linearizations instead of the more straightforward ones
from~\cite{Vrijer1999}. The original proof by Chew~\cite{Chew1981}
uses yet another different but related method, but Chew's proof was
later found to contain a gap.

In general, the methods of the present paper are broadly related to
the problem of establishing the~UN property for classes of term
rewriting systems which include non-left-linear non-confluent
systems. Aside of Chew's theorem, some other work in this direction
has been carried out in
e.g.~\cite{KlopVrijer1989,ToyamaOyamaguchi1995,Verma1997,Stovring2006,KahrsSmith2016}.

In order to increase confidence in the correctness of the main result
of this paper, we have formalized our proof of confluence of~$\CLC$ in
the~Coq proof assistant. The formalization is available
online\footnote{\url{http://www.mimuw.edu.pl/~lukaszcz/clc.tar.gz}}. It
follows closely the development presented here. We used the
CoqHammer~\cite{CzajkaKaliszyk2017Submitted} tool and the automated
reasoning tactics included with it.

\section{Proof overview}\label{sec_overview}

In this section we present an informal overview of the proof, trying
to convey the underlying intuitions. Section~\ref{sec_definitions}
presents formal definitions of the notions informally motivated here,
and Section~\ref{sec_proof} provides details of the proof itself.

We assume familiarity with basic
term-rewriting~\cite{BaaderNipkow1999,Terese2003}. By~$\to^*$ we
denote the transitive-reflexive closure of a relation~$\to$,
by~$\to^\equiv$ its reflexive closure, by~$\leftrightarrow$ the
symmetric closure, and by~$=$ the reflexive-transitive-symmetric
closure. We use~$\equiv$ to denote identity of terms. By~$\to^!$ we
denote reduction to normal form, i.e., $t \to^! s$ if $t \reduces s$
and~$s$ is in normal form. By~$\cdot$ we denote composition of
relations, e.g.~$t \to \cdot \from t'$ holds iff there exists~$t_0$
such that $t \to t_0$ and $t_0 \from t'$. We use the standard notions
of subterms and subterm occurrences, which could be formally defined
by introducing the notion of positions. If~$t$ is a redex,
i.e.~$t \equiv \sigma l$ for some term~$l$ and substitution~$\sigma$,
then a subterm~$s$ \emph{occurs below a variable position} of the
redex~$t$ if~$s$ occurs in a subterm of~$t$ occurring at the position
of a variable in~$l$. The contraction in $t_1 \to t_2$ \emph{occurs at
  the root} if~$t_1$ is the contracted redex.

Let~$u$ be a normal form w.r.t.~a relation~$\to$. The relation~$\to$
(or the underlying rewrite system) is \emph{$u$-normal} if for
every~$t$ such that $t = u$ we have $t \reduces u$.

The most difficult part of our confluence proof is to show that~$\CLC$
is $\Fs$-normal (Lemma~\ref{lem_f_nf}). The confluence of~$\CLC$
(and~$\CLCp$) is then obtained by a relatively simple argument similar
to the one used in~\cite{Vrijer1999} to derive the confluence
of~$\CLCp$ from~$\Ts \ne_{\sCLCz} \Fs$.

An important observation is that $q_1 =_{\sCLC} q_2$ and
$q_1 =_{\sCLCz} q_2$ are in fact equivalent
(Lemma~\ref{lem_clc_equivalent}). Hence, we will use~${=_{\sCLC}}$
and~${=_{\sCLCz}}$ interchangeably. In particular, we actually prove
that for any term~$q$, if $q =_{\sCLCz} \Fs$ then
$q \reduces_{\sCLC} \Fs$.

A naive approach to prove this could be to proceed by induction on the
length of the conversion $q =_{\sCLCz} \Fs$. In the inductive step we
would need to prove:
\begin{enumerate}
\item if $q \reduces_\sCLC \Fs$ and $q \to_{\sCLCz} q'$ then
  $q' \reduces_\sCLC \Fs$,
\item if $q \reduces_\sCLC \Fs$ and
  $q \leftidx{{}_{\sCLCz}}{\from} q'$ then $q' \reduces_\sCLC \Fs$.
\end{enumerate}
The second part is obvious, but the first one is hard. The difficulty
stems from the existence of a non-trivial overlap between the rules
for~$\Cs$. If $t_1 =_\sCLC t_2$ then $\Cs \Fs t_1 t_2 \to_{\sCLC} t_1$
by the third rule of~$\CLC$ and $\Cs \Fs t_1 t_2 \to_{\sCLCz} t_2$ by
the second rule of~$\CLCz$. We do not know enough about~$t_1$
and~$t_2$ to easily infer that they have a common reduct in~$\CLC$.

One may try to strengthen the inductive hypothesis in the hope of
making the first part easier to prove. A naive attempt would be to
claim that \emph{all} reductions starting from~$q$ end in~$\Fs$,
instead of claiming that \emph{some} reduction ends in~$\Fs$. This
would make the first part trivial, but the second one would not go
through as this is false in general, e.g., consider $\Ks \Fs \Omega$
where $\Omega \equiv (\Ss \Is \Is) (\Ss \Is \Is)$ and
$\Is \equiv \Ss \Ks \Ks$.

The idea is to consider, for a given conversion $q =_{\sCLCz} \Fs$, a
certain set $\Sc(q =_{\sCLCz} \Fs)$ of reductions, all starting
from~$q$. The set $\Sc(q =_{\sCLCz} \Fs)$ depends on the exact form of
$q =_{\sCLCz} \Fs$. Then our two parts of the proof for the inductive
step become:
\begin{enumerate}
\item if $\Sc(q =_{\sCLCz} \Fs)$ is nonempty and all reductions in it
  end in~$\Fs$, and $q \to_\sCLC q'$, then
  $\Sc(q' \leftidx{{}_{\sCLCz}}{\from} q =_\sCLCz \Fs)$ is nonempty
  and all reductions in it end in~$\Fs$,
\item if $\Sc(q =_{\sCLCz} \Fs)$ is nonempty and all reductions in it
  end in~$\Fs$, and $q \leftidx{{}_{\sCLCz}}{\from} q'$, then
  $\Sc(q' \to_{\sCLCz} q =_\sCLCz \Fs)$ is nonempty and all reductions
  in it end in~$\Fs$.
\end{enumerate}
The hope is that if we define $\Sc(q =_{\sCLCz} \Fs)$ appropriately,
then showing both parts will become feasible.

Essentially, the set $\Sc(q =_\sCLCz \Fs)$ will be encoded in the
labeling of certain constants in~$q$. The labels determine which
contractions are permitted when a given constant appears as the
leftmost constant in a redex\footnote{E.g.~in the redex
  $\Cs \Ts t_1 t_2$ the constant~$C$ is the leftmost constant.}. At
present the author does not know an ``explicit'' characterization of
the set of reductions $\Sc(q =_\sCLCz \Fs)$ implicitly defined by the
labelings described below.

Terms with the leftmost constant labeled will be called
``significant'', or $s$-terms, whereas others will not contain any
labels and will be called ``insignificant'', or $i$-terms
(c.f.~Definition~\ref{def_terms}). Reductions occurring in $i$-terms
will be ``insignificant'', or $i$-reductions. A ``significant''
contraction, or $s$-contraction, will be a contraction of a term with
the leftmost constant labeled, in a way permitted by the label of the
leftmost constant. Contraction of a redex in which the leftmost
constant is not labeled is not permitted in $s$-contractions. See
Definition~\ref{def_clc_s}. The intuition is that we do not need to
care about the expansions and contractions occurring in
``insignificant'' subterms of a given term, since they cannot
influence the $s$-reductions starting from this term and ending
in~$\Fs$.

The set $\Sc(q =_{\sCLCz} \Fs)$ will be encoded in a labeled
variant\footnote{By a ``labeled variant'' of a term~$q$ we mean a term
  with certain constants labeled which is identical with~$q$ when the
  labels are ``erased''.}~$t$ of~$q$, and it will consist of all
$s$-reductions starting from~$t$ and ending in a normal form
(w.r.t.~$s$-contraction). Strictly speaking, we have just silently
shifted from considering contractions in ``plain'' terms of the
system~$\CLC$ to contractions in their labeled variants, in a
different rewriting system which we have not yet defined. In
particular, we will actually be interested in $s$-reductions ending in
a labeled variant~$\Fs_1$ of~$\Fs$. However, it will be later shown
that $s$-reductions defined on labeled terms may be ``erased'' to
appropriate reductions in the system~$\CLC$. In the next section we
define the system~$\CLCs$ (Definition~\ref{def_clc_s}) over labeled
terms (Definition~\ref{def_terms}) which will give precise rules of
$s$-contraction. In this section we only give informal motivations.

The labels constrain the ways in which $s$-redexes may be contracted
and encode permissible $s$-reductions to~$\Fs_1$. Each term decomposes
into a ``significant'' prefix and an ``insignificant'' suffix
(c.f.~\ref{std_i_or_s_or_tuple} in Definition~\ref{def_standard}). The
``significant'' prefix contains all labeled constants and no unlabeled
constants. The ``insignificant'' suffix consists of all
``insignificant'' subterms. All constants in the ``insignificant''
suffix are unlabeled. This is analogous to the existence of a needed
prefix and a non-needed suffix in orthogonal
TRSs~\cite[Section~9.2.2]{Terese2003}. An ``insignificant'' subterm
does not overlap with any needed redexes. In particular, it does not
contain any needed redexes. No position inside an ``insignificant''
subterm (dynamically) traces to~$\Fs_1$ along any $s$-reduction
to~$\Fs_1$ (c.f.~\cite[Definition~8.6.7]{Terese2003}). In contrast,
each $s$-redex needs to be either $s$-contracted or erased by a rule
for~$\Cs_2$ (see Definition~\ref{def_clc_s}) in any $s$-reduction
to~$\Fs_1$. Each position of a labeled constant either traces
to~$\Fs_1$ along a given $s$-reduction to~$\Fs_1$, or is erased in
that $s$-reduction by a rule for~$\Cs_2$. An $s$-reduct of an $s$-term
is always also an $s$-term (c.f.~\ref{std_s_term_reduce} in
Definition~\ref{def_standard}).

We write $t \to_s t'$ for one-step reduction in~$\CLCs$. We use the
abbreviation $s$-NF for $\CLCs$-normal form. We write $t \Da_{\Fs_1}$
when, among other conditions to be defined later, $t$ is complete,
i.e.~terminating and confluent, w.r.t.~$s$-reductions with~$\Fs_1$ as
the normal form (c.f.~Definition~\ref{def_standard}).

With the set $\Sc(q =_{\sCLCz} \Fs)$ coded by labels, the two parts of
the inductive step become:
\begin{enumerate}
\item if~$t$ is a labeled variant of~$q$ such that $t \Da_{\Fs_1}$,
  and $q \to_{\sCLCz} q'$, then there exists a labeled variant~$t'$
  of~$q'$ such that $t' \Da_{\Fs_1}$ (c.f.~Corollary~\ref{cor_contr}),
\item if~$t$ is a labeled variant of~$q$ such that $t \Da_{\Fs_1}$,
  and $q \leftidx{{}_{\sCLCz}}{\from} q'$, then there exists a labeled
  variant~$t'$ of~$q'$ such that $t' \Da_{\Fs_1}$
  (c.f.~Corollary~\ref{cor_expand}).
\end{enumerate}

Now we provide some explanations on how the terms will be labeled. For
this purpose we analyze why the second part fails when we take
$\Sc(q =_{\sCLCz} F)$ to be the set of all reductions starting
from~$q$. We indicate how to introduce the labeled variants so as to
make the second part go through while still retaining the feasibility
of showing the first part.

Suppose $q \leftidx{{}_{\sCLCz}}{\from} q'$ at the root and we have
already decided on the labeled variant~$t$ of~$q$. We need to decide
on a labeled variant~$t'$ of~$q'$, and assign appropriate meaning to
the labels, in such a way that the second part goes through. In short,
in~$t'$ we preserve the labelings of the subterms of~$q'$ which are
copied to~$q$ in $q' \to_{\sCLCz} q$, we do not label the new subterms
of~$q'$ which are erased in $q' \to_{\sCLCz} q$ (they become
$i$-terms), and we ensure that $i$-terms and cannot influence
\emph{any} $s$-reduction from~$t'$ to~$\Fs_1$. First of all, if~$t$ is
an $i$-term, i.e., $t \equiv q$, then we may take $t' \equiv q'$. So
assume~$t$ is an $s$-term. Then there are the following possibilities.
\begin{itemize}
\item If $q' \equiv \Cs \Ts q q_0 \to_{\sCLCz} q$ then~$q_0$ is a new
  subterm. We take $t' \equiv \Cs_1 \Ts_1 t q_0$. The labeling~$\Cs_1$
  of~$\Cs$ will be interpreted as not permitting contraction by the
  third rule, i.e., in~$\CLCs$ we will only have the rules
  $\Cs_1 \Ts_1 x y \to x$ and $\Cs_1 \Fs_1 x y \to y$. This ensures
  that~$q_0$ gets erased in every $s$-reduction of~$t'$ to~$\Fs_1$.
\item The case when $q' \equiv \Cs \Fs q_0 q \to_{\sCLCz} q$ is
  analogous: we take $t' \equiv \Cs_1 \Fs_1 q_0 t$.
\item If $q' \equiv \Cs q_0 q q \to_{\sCLCz} q$ by the third rule,
  then~$q_0$ is a new term. We take $t' \equiv \Cs_2 q_0 t t$. In the
  system~$\CLCs$ we have two rules for~$\Cs_2$
  \[
  \begin{array}{rcl}
    \Cs_2 z x y &\to& x \quad\From\quad \erase{x} =_{\sCLC} \erase{y} \\
    \Cs_2 z x y &\to& y \quad\From\quad \erase{x} =_{\sCLC} \erase{y}
  \end{array}
  \]
  where $\erase{x} =_\sCLC \erase{y}$ means that the ``erasures'' of
  the labeled terms substituted for~$x$ and~$y$ must be equal
  in~$\CLC$ for the rule to be applicable. These rules ensure
  that~$q_0$ cannot influence any $s$-reduction of~$t'$ to~$\Fs_1$ --
  it gets erased in each.

  The presence of the second rule for~$\CLCs$ is not a problem,
  because we will only consider terms terminating in~$\CLCs$. Whenever
  the second rule is applicable, so is the first one, hence if all
  maximal $s$-reductions end in~$\Fs_1$, then there is an
  $s$-reduction ending in~$\Fs_1$ which does not use the second rule
  for~$\Cs_2$ (Lemma~\ref{lem_good_reduction}). It will be easy to
  ``erase'' an $s$-reduction not using the second rule for~$\Cs_2$ to
  obtain a reduction in~$\CLC$ (Lemma~\ref{lem_s_minus_erase}).
\item If $q' \equiv \Ks q q_0 \to_\sCLC q$ then we take
  $t' \equiv \Ks_1 t q_0$. The rule for~$\Ks_1$ in~$\CLCs$ is
  $\Ks_1 x y \to x$.
\item If
  $q' \equiv \Ss q_1 q_2 q_3 \to_\sCLC q_1 q_3 (q_2 q_3) \equiv q$
  then we run into a problem with our labeling approach, because the
  labeled variants of the distinct occurrences of~$q_3$ may be
  distinct. Suppose~$t_1$ is the labeled variant of~$q_1$, the
  term~$t_2$ of~$q_2$, the term~$t_3$ of the first~$q_3$, and~$t_3'$
  of the second~$q_3$. We cannot just arbitrarily choose e.g.~$t_3$
  and say that $\Ss_1 t_1 t_2 t_3$ is the labeled variant of~$q'$,
  because contracting $\Ss_1 t_1 t_2 t_3$ yields $t_1 t_3 (t_2 t_3)$,
  not $t_1 t_3 (t_2 t_3')$, and now the second occurrence of~$q_3$ has
  the wrong labeling.

  A solution is to remember both labeled variants of~$q_3$. So the
  labeled variant of~$q'$ would be
  e.g.~$\Ss_1 t_1 t_2 \la t_3, t_3' \ra$. In~$\CLCs$ the rule
  for~$\Ss_1$ would be
  \[
    \Ss_1 x_1 x_2 \la x_3, x_3' \ra \to x_1 x_3 (x_2 x_3').
  \]
  However, once we introduce such pairs, terms of the form
  $\Ss_1 t_1 t_2 \la t_3, t_3' \ra$ may appear in the terms being
  expanded. This is not a problem for any of the rules of~$\CLC$
  except the rule for~$\Ss$, because the right sides of all other
  rules are variables.

  Consider for instance
  $q \equiv q_0 q_3 (\Ss q_1 q_2 q_3) \leftidx{{}_\sCLC}{\from} \Ss
  q_0 (\Ss q_1 q_2) q_3 \equiv q'$. Suppose
  $t \equiv t_0 t_3 (\Ss_1 t_1 t_2 \la t_3, t_3' \ra)$. Now the
  term~$q_3$ has three possibly distinct labeled variants, and we need
  to remember all of them in a tuple. We will thus introduce a new
  labeling of~$\Ss$ for every possible labeling of the right
  side~$x z (y z)$ of the rule for~$\Ss$ in the system~$\CLC$.

  By introducing the tuples in the labelings we in essence put
  constraints on the order in which $s$-redexes may be contracted
  (think of all reductions inside a tuple as ``really'' occuring after
  the surrounding $\Ss$-redex is contracted). At present the author
  does not know a precise ``explicit'' characterization of these
  constraints.
\end{itemize}
Note that by labeling $\Cs$ differently in $\Cs \Fs q_1 q_2$ and
$\Cs q_0 q q$ we effectively eliminated in~$\CLCs$ the problematic
non-trivial overlap occuring in~$\CLC$. Now a new ``insignificant''
term created in an expansion cannot later on appear in place of a
``significant'' term as a result of an ``incompatible'' contraction. A
redex inside an ``insignificant'' subterm cannot suddenly become
needed in an $s$-reduction -- it is erased in any $s$-reduction to
normal form.

We also need to ensure that we can handle the first part of the
inductive step when $q \to_{\sCLCz} q'$. Suppose~$t$ is the labeled
variant of~$q$. We need to find a labeled variant for~$q'$. For
simplicity assume that there is only one position in~$t$ which
corresponds to the position of the contraction in~$q$. If the
contraction occurs inside an $i$-term in~$t$, then it does not matter
and we may label~$q'$ in the same way as~$q$. If an $s$-term is
contracted in a way permitted for significant contraction, then it is
also obvious how to label~$q'$ -- just take the labeled variant
of~$q'$ to be the reduct of the labeled variant of~$q$. But what if
neither of the two holds?

For instance, what if $t \equiv \Cs_1 t_0 t_1 t_2$ but
$q \equiv \Cs q_0 q' q' \to_\sCLC q'$? This possibility is not
problematic, provided that~$t_0 \to_s^* \Ts_1$ or $t_0 \to_s^* \Fs_1$,
which will be the case because~$t_0$ was ``obtained'' from~$\Ts_1$
or~$\Fs_1$ by a conversion with the intermediate terms labeled
appropriately (c.f.~\ref{std_c_1} in Definition~\ref{def_standard}
and~\ref{prop_std_c_1} in Lemma~\ref{lem_standard_properties}). If
e.g.~$t_0 \to_s^* \Ts_1$ then we take~$t_1$ to be the labeling
of~$q'$. We then have $\Cs_1 t_0 t_1 t_2 \to_s^* \Cs_1 \Ts_1 t_1 t_2$
and the contraction $\Cs_1 \Ts_1 t_1 t_2 \to_s t_1$ is permitted for
``significant'' contractions.

The last problematic case is when e.g.~$t \equiv \Cs_2 \Fs t_1 t_2$ is
the labeling of~$q \equiv \Cs \Fs q_1 q'$, and $q \to_{\sCLCz} q'$ by
the second rule. However, because $\Cs_2 \Fs t_1 t_2$ was ``obtained''
from $\Cs_2 q t' t'$ we will have $\erase{t_1} =_\sCLC \erase{t_2}$
(c.f.~\ref{std_c_2} in Definition~\ref{def_standard}). Then the second
rule for~$\Cs_2$ in~$\CLCs$ is applicable and we may take~$t_2$ as the
labeling of~$q'$.

\section{Definitions}\label{sec_definitions}

This section is devoted to fixing notation and introducing definitions
of various technical concepts. First, we clarify the formal definition
of conditional term rewriting systems. For more background on
conditional rewriting see e.g.~\cite{Terese2003}.

\begin{definition}\label{def_trs}
  A \emph{conditional rewrite rule} is a rule of the form
  $l \to r \From P(x_1,\ldots,x_n)$, where~$l$ is not a variable,
  $\Var(r) \subseteq \Var(l)$, $x_1,\ldots,x_n\in\Var(l)$,
  and~$P(x_1,\ldots,x_n)$ is the \emph{condition} of the rule,
  with~$P$ a fixed predicate on terms. The predicate~$P$ may refer to
  the conversion relation~$=$ of the conditional term rewriting system
  being defined. A term~$t$ is a \emph{redex} (\emph{contractum}) by
  this rule if there is a substitution~$\sigma$ such that
  $t \equiv \sigma l$ ($t \equiv \sigma r$) and
  $P(\sigma(x_1),\ldots,\sigma(x_n))$ holds. A \emph{conditional term
    rewriting system}~$R$ is a set of conditional rewrite
  rules. Because the conditions in the rules may refer to the
  conversion relation of~$R$, the definition is circular. Formally, an
  $R$-contraction $q \contr_R q'$ is defined in the following
  way. Define~$R_0$ to be the system~$R$ but using the equality
  relation in place of~$=$ in the conditions, and~$R_{n+1}$ to be the
  system~$R$ with the conversion relation~$=_{R_n}$ of~$R_n$ used in
  place of~$=$. We then define $q \contr_R q'$ to hold if there
  is~$n \in \Nbb$ with $q \contr_{R_n} q'$. The least such~$n$ is
  called the \emph{level} of the contraction. If the conditions are
  continuous w.r.t.~$=$ then the relation~$\contr_R$ is a fixpoint of
  the above construction, i.e., it is the contraction relation of the
  system~$R_\infty$ which uses~$=_R$ in place of~$=$. Let~$\sim$ be a
  binary relation on terms. If for any substitution~$\sigma$ such
  that~$P(\sigma(x_1),\ldots,\sigma(x_n))$ holds, and any~$\sigma'$
  such that $\sigma(x) \sim \sigma'(x)$ for all variables~$x$,
  also~$P(\sigma'(x_1),\ldots,\sigma(x_1'))$ holds, then the
  condition~$P(x_1,\ldots,x_n)$ is \emph{stable under~$\sim$}.
\end{definition}

The following is a simple but crucial observation, which implies that
it suffices to consider conversions in~$\CLCz$. A generalization of
this fact was already shown in~\cite[Lemma~3.7]{Vrijer1999}. The proof
is by induction on the maximum level of the contractions/expansions in
$q =_{\sCLCp} q'$.

\begin{lemma}\label{lem_clc_equivalent}
  The following are equivalent: $q =_{\sCLCz} q'$, $q =_{\sCLC} q'$,
  and $q =_{\sCLCp} q'$.
\end{lemma}

\begin{definition}\label{def_terms}
  We define \emph{insignificant terms}, or \emph{$i$-terms}, to be the
  terms of~$\CLC$, i.e., terms over the
  signature~$\Sigma = \{@, \Cs, \Ts, \Fs, \Ks, \Ss \}$ where~$@$ is a
  binary function symbol and the other symbols are constants. We write
  $t_1 t_2$ instead of $@(t_1,t_2)$.  The set of \emph{labeled terms},
  or \emph{$l$-terms}, is the set of terms over the signature
  consisting of the symbols of~$\Sigma$, the \emph{labeled constants}
  $\Cs_1, \Cs_2, \Ts_1, \Fs_1, \Ks_1$ and $\Ss^{n_0,\ldots,n_k}$ for
  each $k,n_1,\ldots,n_k \in \Nbb_+$, and an $n$-ary function
  symbol~$\Ps^n$ for each $n \in \Nbb_+$. We write
  $\la t_1,\ldots,t_n \ra$ instead of $\Ps^n(t_1,\ldots,t_n)$. We
  adopt the convention $\la t \ra \equiv t$. If
  $t \equiv \la t_1, \ldots, t_n \ra$ with $n > 1$, then we say
  that~$t$ is a \emph{tuple of length~$n$}. Note that
  $\la t \ra \equiv t$ is just a notational convention. We say that
  $\la t_1, \ldots, t_n \ra$ is a tuple \emph{only when $n > 1$}.

  An \emph{erasure} of an $l$-term is defined as follows:
  \begin{itemize}
  \item an $i$-term is an erasure of itself,
  \item $\Cs$ is an erasure of $\Cs_1$ and~$\Cs_2$; $\Ts$ is an erasure of
    $\Ts_1$; $\Fs$ is an erasure of $\Fs_1$; $\Ks$ is an erasure of $\Ks_1$; $\Ss$
    is an erasure of $\Ss^{n_1,\ldots,n_k}$,
  \item if $q_1$, $q_2$ are erasures of $t_1$, $t_2$, respectively,
    then~$q_1 q_2$ is an erasure of~$t_1 t_2$,
  \item if $q_i$ is an erasure of $t_i$, for some $1 \le i \le n$,
    then~$q_i$ is an erasure of $\la t_1,\ldots,t_n \ra$.
  \end{itemize}
  The \emph{leftmost erasure} of~$t$, denoted~$\erase{t}$, is the
  erasure in which we always choose~$i=1$ in the last point above. We
  write $t \equiverased q$ if \emph{every} erasure of~$t$ is identical
  with~$q$.

  We define \emph{significant terms}, or \emph{$s$-terms},
  inductively.
  \begin{itemize}
  \item Any labeled constant is an $s$-term.
  \item If $t_1$ is an $s$-term and $t_2$ is an $l$-term, then $t_1
    t_2$ is an $s$-term.
  \end{itemize}
  In other words, an $s$-term is an $l$-term whose leftmost constant
  is labeled.
\end{definition}

In what follows $t$, $t_1$, $t_2$, $r$, $r_1$, $r_2$, $s$, $s_1$,
etc.~stand for $l$-terms; and $q$, $q_1$, $q_2$, etc.~stand for
$i$-terms; unless otherwise qualified. Also, whenever we talk about
terms without further qualification, we implicitly assume them to be
$l$-terms.

\begin{definition}\label{def_clc_s}
  The system~$\CLCs$ is defined by the following \emph{significant
    reduction rules}:
  \[
  \begin{array}{rclcrcl}
    \Cs_1 \Ts_1 x y &\to& x &\quad\quad& \Cs_2 z x y &\to& x \quad\From\quad \erase{x} =_{\sCLC} \erase{y} \\
    \Cs_1 \Fs_1 x y &\to& y &\quad\quad& \Cs_2 z x y &\to& y \quad\From\quad \erase{x} =_{\sCLC} \erase{y} \\
    \Ks_1 x y &\to& x &\quad& &&
  \end{array}
  \]
  \[
  \begin{array}{l}
    \Ss^{\vec{n}} x \la y_1,\ldots,y_{k}\ra \la \vec{z}_0, \ldots, \vec{z}_k
    \ra\quad
    \to
    \quad
    x \la \vec{z}_0 \ra \la (y_1 \la \vec{z}_1 \ra),
    \ldots, (y_{k} \la \vec{z}_{k} \ra) \ra
    \quad\From \quad \varphi
  \end{array}
  \]
  where
  \begin{eqnarray*}
    \varphi &\equiv& \erase{z_{i,j}} =_{\sCLC}
    \erase{z_{i',j'}} \text{ for } i,i'=0,\ldots,k,\,
    j=1,\ldots,n_i,\, j'=1,\ldots,n_{i'}, \text{ and} \\ && \erase{y_i} =_{\sCLC}
    \erase{y_j} \text{ for } i,j=1,\ldots,k,
  \end{eqnarray*}
  and $\vec{n}$ stands for $n_0,\ldots,n_{k}$, and $\vec{z}_i$ stands
  for $z_{i,1},\ldots,z_{i,n_i}$, for $i=0,\ldots,k$. When dealing
  with terms whose leftmost constant is~$\Ss^{n_0,\ldots,n_k}$, we
  will often use this kind of vector notation. Recall the convention
  $\la t \ra \equiv t$. Hence, if e.g.~$n_0 = 1$, then
  $\la \vec{z}_0 \ra \equiv \la z_{0,1} \ra \equiv z_{0,1}$ in the
  above rule. The condition~$\varphi$ ensures that the leftmost
  erasures of all~$z_{i,j}$ are convertible in~$\CLC$, and that the
  leftmost erasures of all~$y_i$ are convertible in~$\CLC$. Some
  examples of significant reduction rules for~$\Ss^{\vec{n}}$
  (omitting the conditions) are:
  \[
    \begin{array}{rcl}
      \Ss^{1,1} x y_1 \la z_{0,1}, z_{1,1}\ra &\to& x z_{0,1} (y_1
                                                    z_{1,1}) \\
      \Ss^{1,2,1} x \la y_1, y_2 \ra \la z_{0,1}, z_{1,1}, z_{1,2},
      z_{2,1}\ra &\to& x z_{0,1} \la y_1 \la z_{1,1}, z_{1,2} \ra, y_2
                       z_{1,2}\ra \\
      \Ss^{2,2} x y_1 \la z_{0,1}, z_{0,2}, z_{1,1}, z_{1,2}\ra &\to& x\la z_{0,1}, z_{0,2}
                                                                      \ra
                                                                      (y_1
                                                                      \la
                                                                      z_{1,1},
                                                                      z_{1,2}
                                                                      \ra)
    \end{array}
  \]
  For instance, the condition for the second of these rules states
  that $\erase{y_1} =_{\sCLC} \erase{y_2}$,
  $\erase{z_{0,1}} =_{\sCLC} \erase{z_{1,1}}$,
  $\erase{z_{0,1}} =_{\sCLC} \erase{z_{1,2}}$,
  $\erase{z_{0,1}} =_{\sCLC} \erase{z_{2,1}}$, $\erase{z_{1,1}}
  =_{\sCLC} \erase{z_{1,2}}$, etc.

  Note that the equality~$=_\sCLC$ in the conditions refers to the
  system~$\CLC$, not~$\CLCs$. Note also that all rules of~$\CLCs$ are
  linear, disregarding the side-conditions.

  Reduction by a rule in~$\CLCs$ is called \emph{significant
    reduction}, or $s$-reduction. One-step $s$-reduction is denoted
  by~$\to_s$. Analogously, we use the terminology and notation of
  $s$-contraction, $s$-expansion, $s$-redex, $s$-normal form ($s$-NF),
  etc. Note that every $s$-redex is an $s$-term. We write
  $t \to_{s{-}} t'$ if $t \to_s t'$ and the $s$-contraction is not by
  the second rule for~$\Cs_2$ and it does not occur inside a tuple.

  An \emph{$i$-redex} is a $\CLC$-redex which is also an $i$-term. An
  $l$-term~$t_1$ is said to \emph{$i$-reduce} to~$t_2$,
  denoted~$t_1 \to_i t_2$, if~$t_1 \to_\sCLC t_2$ and the redex
  contracted in~$t_1$ is an $i$-term. An $l$-term~$t_1$ is said to
  \emph{$i$-expand} to~$t_2$ if $t_2 \to_i t_1$. We write
  $t_1 \to_{i,s} t_2$ if $t_1 \to_i t_2$ or $t_1 \to_s t_2$.
\end{definition}

Actually, we will consider mostly $l$-terms whose all erasures are
identical. For such a term an $s$-contraction by a rule
for~$\Ss^{\vec{n}}$ in~$\CLCs$ naturally corresponds to a
$\CLC$-contraction on its erasure. We could get rid of the side
conditions in the rules for~$\Ss^{\vec{n}}$ and consider exclusively
terms whose all erasures are identical. But then we would need to
require $i$/$s$-contractions/expansions to always occur ``in the same
way'' (modulo labeling) in all components of a tuple. This would
complicate the inductive proofs concerning the relations~$\to_s$,
$\to_i$, etc. Hence, the role of the conditions in the rules
for~$\Ss^{\vec{n}}$ is purely technical.

\begin{lemma}\label{lem_sn}
  The system $\CLCs$ is terminating.
\end{lemma}

\begin{proof}
  The number of labeled constants decreases with each
  $s$-contraction.
\end{proof}

\begin{lemma}
  If $t_1 \to_s t_2$ then $\erase{t_1} =_\sCLC \erase{t_2}$.
\end{lemma}

The above simple lemma implies that the conditions in significant
reduction rules are stable under $s$-reduction and $s$-expansion. It
is obvious that they are also stable under $i$-reduction and
$i$-expansion.

\begin{lemma}\label{lem_s_minus_erase}
  If $t \equiverased q$ and $t \to_{s{-}} t'$ then there is~$q'$ with
  $q \to_{\sCLC} q'$ and $t' \equiverased q'$.
\end{lemma}

\begin{proof}
  Because all erasures of~$t$ are identical and the second rule
  for~$\Cs_2$ is not used, the $s{-}$-reduction may be simulated by a
  $\CLC$-reduction in an obvious way. Because the $s{-}$-contraction
  does not occur inside a tuple, all erasures of~$t'$ are still
  identical.
\end{proof}

In the next definition we introduce the predicate~$\Da_{\Fs_1}$ and
the notion of standard $l$-terms. Intuitively, an $l$-term~$t$ is
standard if the labelings in~$t$ have the meaning we intend to assign
them, i.e.~if~$t$ is a term obtained by the process informally
described in the previous section.

\begin{definition}\label{def_standard}
  An $l$-term~$t$ is \emph{standard} if for every subterm~$t'$ of~$t$
  the following hold:
  \begin{enumerate}
  \item\label{std_i_or_s_or_tuple} $t'$ is either an $i$-term, an
    $s$-term or a tuple,
  \item\label{std_c_1} if $t' \equiv \Cs_1 t_0 t_1 t_2$ and $t_0$ is in
    $s$-NF, then $t_0 \equiv \Ts_1$ or $t_0 \equiv \Fs_1$,
  \item\label{std_c_2} if $t' \equiv \Cs_2 t_0 t_1 t_2$ then
    $\erase{t_1} =_{\sCLC} \erase{t_2}$,
  \item\label{std_s} if $t' \equiv \Ss^{n_0,\ldots,n_k} t_0 t_1 t_2$
    then~$t_2$ is a tuple of length $\sum_{i=0}^k n_k$ and if $k > 1$
    then~$t_1$ is a tuple of length~$k$,
  \item\label{std_s_term_reduce} if $t'$ is an $s$-term and
    $t' \to_s^* t''$, then~$t''$ is also an $s$-term,
  \item\label{std_no_nested_tuple} if
    $t' \equiv \la t_1,\ldots,t_n \ra$ with $n > 1$, then none of
    $t_1,\ldots,t_n$ is a tuple.
  \end{enumerate}
  An $l$-term~$t$ is \emph{strongly standard} if $t \to_s^* t'$
  implies that~$t'$ is standard. We write $t\Da_{\Fs_1}$ if~$t$ is
  strongly standard and has no $s$-NFs other than~$\Fs_1$, i.e.~if
  $t \to_s^! t'$ then $t' \equiv \Fs_1$.
\end{definition}

Point~\ref{std_i_or_s_or_tuple} in Definition~\ref{def_standard}
essentially ensures that a standard term may be decomposed into a
``significant'' prefix and an ``insignificant'' suffix. A labeled term
which is not standard is e.g.~$\Cs \Ts_1$, because it is neither an
$s$-term, nor an $i$-term, nor a tuple. Other examples of non-standard
terms are: $\Cs_1 \Ts \Fs \Fs$, $\Cs_2 \Cs \Ts \Fs$,
$\Ss^{1,1,1} \Ts \Ts \Ts$, $\Ks_1 \Fs \Fs$, $\Cs_2 \Cs \Ts_1 \Ts$,
$\Ss^{1,1} \Cs \Cs \la \Ts_1, \Ts_1 \ra$,
$\Ks_1 \la \Ts_1, \Ts_1 \ra \Ts_1$, $\la \la \Cs, \Cs \ra, \Cs
\ra$. Examples of standard terms which are not strongly standard are:
$\Ss^{1,1} \Cs_1 \Cs \la \Ts_1, \Ts_1 \ra$,
$\Ss^{1,1} \Cs_1 \Cs \la \Ts, \Ts \ra \Ts$,
$(\Ks_1 \Cs_1 \Ts) \Ts \Ts \Ts$.

\begin{lemma}\label{lem_standard_properties}
  \begin{enumerate}
  \item Any $i$-term is standard.
  \item Any labeled constant is standard.
  \item Every subterm of a standard term is also standard.
  \item Every subterm of a term to which some strongly standard term
    $s$-reduces, is strongly standard.
  \item If $t_1 t_2$ is standard then $t_1$ is not a tuple.
  \item\label{prop_std_c_1} If $\Cs_1 t_0 t_1 t_2$ is a subterm of a
    strongly standard term, then $t_0 \to_s^* \Ts_1$ or
    $t_0 \to_s^* \Fs_1$.
  \end{enumerate}
\end{lemma}

\begin{proof}
  Follows from definitions. For the last point one also needs
  Lemma~\ref{lem_sn}.
\end{proof}

\section{Confluence proof}\label{sec_proof}

We now give technical details of our confluence proof. As outlined in
Section~\ref{sec_overview}, we show:
\begin{enumerate}
\item if $t \succ q$ and $t \Da_{\Fs_1}$, and $q \to_{\sCLCz} q'$,
  then there is~$t'$ with $t' \succ q'$ and $t' \Da_{\Fs_1}$
  (Corollary~\ref{cor_contr}),
\item if $t \succ q$ and $t \Da_{\Fs_1}$, and
  $q \leftidx{{}_{\sCLCz}}{\from} q'$, then there is~$t'$ with
  $t' \succ q'$ and $t' \Da_{\Fs_1}$ (Corollary~\ref{cor_expand}).
\end{enumerate}
The first part is proven by showing that $\CLCz$-reductions in~$q$ may
be simulated by $i$-reductions and $s$-reductions in~$t$, and that
$i$/$s$-reductions preserve~$\Da_{\Fs_1}$ (Lemma~\ref{lem_eqv_i} and
Lemma~\ref{lem_contr_s}). For the second part, we show that
$\CLCz$-expansions in~$q$ may be simulated by $i$-expansions and
$a$-expansions (Definition~\ref{def_a_redex}) in~$t$. The technical
notion of $a$-expansion is needed to ensure that the new subterms
of~$t'$ are labeled appropriately, in the way outlined in
Section~\ref{sec_overview} ($s$-contraction by itself does not put any
labeling restrictions on the terms erased in the
contraction). Moreover, $a$-expansion is also needed to facilitate the
proof that $t'\Da_{\Fs_1}$ (see the discussion before
Definition~\ref{def_a_redex}). Plain $s$-expansion does not
necessarily preserve~$\Da_{\Fs_1}$, while $a$-expansion does
(Lemma~\ref{lem_expand_a}).

In other words, we show that $\CLCz$-reductions (expansions) in
unlabeled terms may be simluated by $i$/$s$-reductions
($i$/$a$-expansions) in their labeled variants, and that
$i$/$s$-reductions ($i$/$a$-expansions) preserve~$\Da_{\Fs_1}$. A
conversion $q =_{\sCLCz} \Fs$ can then be translated into a conversion
$t =_{i,s,a} \Fs_1$ with no $s$-expansions or $a$-reductions, and
with~$t \equiverased q$. For instance, a conversion in~$\CLCz$
\[
\begin{array}{l}
  \Fs \from \Cs (\Ks \Fs \Omega) \Fs \Fs \from \Cs (\Ks \Fs \Omega)
  \Fs (\Cs \Ts \Fs (\Ks \Fs \Omega)) \to \Cs \Fs \Fs (\Cs \Ts \Fs (\Ks
  \Fs \Omega)) \from \\ \quad \Cs \Fs \Fs (\Cs (\Ks \Ts \Fs) \Fs (\Ks \Fs
  \Omega)) \to \Cs (\Ks \Ts \Fs) \Fs (\Ks \Fs \Omega) \to \Cs (\Ks \Ts
  \Fs) \Fs \Fs \to \Fs \from \\ \quad \Ks \Fs (\Cs \Fs) \from \Ss \Ks \Cs \Fs
  \from \Ss \Ks \Cs (\Ks \Fs \Omega) \to \Ss \Ks \Cs \Fs
\end{array}
\]
will be translated to
\[
\begin{array}{l}
  \Fs_1 \leftidx{{}_a}{\from} \Cs_2 (\Ks \Fs \Omega) \Fs_1 \Fs_1 \leftidx{{}_a}{\from} \Cs_2 (\Ks \Fs \Omega)
  \Fs_1 (\Cs_1 \Ts_1 \Fs_1 (\Ks \Fs \Omega)) \to_i \Cs_2 \Fs \Fs_1 (\Cs_1 \Ts_1 \Fs_1 (\Ks
  \Fs \Omega)) \leftidx{{}_a}{\from} \\ \quad \Cs_2 \Fs \Fs_1 (\Cs_1 (\Ks_1 \Ts_1 \Fs) \Fs_1 (\Ks \Fs
  \Omega)) \to_s \Cs_1 (\Ks_1 \Ts_1 \Fs) \Fs_1 (\Ks \Fs \Omega) \to_i \Cs_1 (\Ks_1 \Ts_1
  \Fs) \Fs_1 \Fs \reduces_s \Fs_1 \leftidx{{}_a}{\from} \\ \quad \Ks_1
  \Fs_1 (\Cs \Fs) \leftidx{{}_a}{\from} \Ss^{1,1} \Ks_1 \Cs \la \Fs_1,
  \Fs \ra
  \leftidx{{}_{a,i}^{~*}}{\from} \Ss^{1,1} \Ks_1 \Cs \la \Ks_1 \Fs_1 \Omega,
  \Ks \Fs \Omega \ra \reduces_{s,i} \Ss^{1,1} \Ks_1 \Cs \la \Fs_1, \Fs \ra
\end{array}
\]
Since~$\Fs_1 \Da_{\Fs_1}$ and we prove that $i$/$s$-reductions and
$i$/$a$-expansions preserve~$\Da_{\Fs_1}$, we may conclude
that~$t \Da_{\Fs_1}$. Then by the definition of~$\Da_{\Fs_1}$ we
obtain a significant reduction $t \to_s^* \Fs_1$. In fact, the
reduction may be assumed to be a $s{-}$-reduction
(Lemma~\ref{lem_good_reduction}). By Lemma~\ref{lem_s_minus_erase}
this reduction $t \to_{s{-}}^{*} \Fs_1$ may be translated into a
$\CLC$-reduction by erasing the labelings. Hence finally
$q \reduces_\sCLC \Fs$ (Lemma~\ref{lem_f_nf}).

We first show that a $\CLCz$-contraction may be simulated by
$i$-reductions and $s$-reductions.

\begin{lemma}\label{lem_erase_contr}
  If~$t$ is strongly standard, $t \equiverased q$ and
  $q \contr_{\sCLCz} q'$, then there exists a term~$t'$ such that
  $t \to_{i,s}^{*} t'$ and $t' \equiverased q'$.
\end{lemma}

\begin{proof}
  Induction on the size of~$t$. First assume~$t$ is not a tuple
  and~$q$ is the $\CLCz$-redex contracted in $q \contr_{\sCLCz}
  q'$. If~$t \equiv q$ then $t \equiv q \to_i q'$ and we may take
  $t' \equiv q'$. If~$t \not\equiv q$ then~$t$ is not an $i$-term
  because $t \equiverased q$. Hence by~\ref{std_i_or_s_or_tuple} in
  Definition~\ref{def_standard} we conclude that~$t$ is an
  $s$-term. We have the following possibilities.
  \begin{itemize}
  \item If $q \equiv \Cs \Ts q_1 q_2 \to_{\sCLCz} q_1 \equiv q'$ then the
    leftmost constant in~$t$ is either~$\Cs_1$ or~$\Cs_2$.
    \begin{itemize}
    \item If $t \equiv \Cs_1 \Ts_1 t_1 t_2$ then $t \to_s t_1$ and
      $t_1 \equiverased q_1$, so we may take $t' \equiv t_1$.
    \item The case $t \equiv \Cs_1 \Ts t_1 t_2$ is impossible
      by~\ref{std_c_1} in Definition~\ref{def_standard}.
    \item If $t \equiv \Cs_2 t_0 t_1 t_2$ then $t_1 \equiverased q_1$
      and $\erase{t_1} =_{\sCLC} \erase{t_2}$ by~\ref{std_c_2} in
      Definition~\ref{def_standard}. Thus $t \to_s t_1$ and we may
      take $t' \equiv t_1$.
    \end{itemize}
  \item If $q \equiv \Cs \Fs q_1 q_2 \to_{\sCLCz} q_2$ then the
    argument is analogous. Note that the presence of the second rule
    for~$\Cs_2$ is necessary here.
  \item If $q \equiv \Cs q_0 q_1 q_1 \to_{\sCLCz} q_1$ then $t \equiv \Cs'
    t_0 t_1 t_2$ with $\Cs' \in \{\Cs_1,\Cs_2\}$, $t_0 \equiverased q_0$,
    $t_1 \equiverased q_1$ and $t_2 \equiverased q_1$.
    \begin{itemize}
    \item If $\Cs' \equiv \Cs_1$ then $t_0 \to_s^{*} \Ts_1$ or $t_0
      \to_s^{*} \Fs_1$ by Lemma~\ref{lem_standard_properties}. Hence $t
      \to_s^{*} t_1$ or $t \to_s^{*} t_2$. In the first case we may
      take $t' \equiv t_1$, and in the second we take $t' \equiv t_2$.
    \item If $\Cs' \equiv \Cs_2$ then $t \to_s t_1$ because
      $\erase{t_1} \equiv \erase{t_2} \equiv q_1$. Thus we take $t'
      \equiv t_1$.
    \end{itemize}
  \item If $q \equiv \Ks q_1 q_2 \to_{\sCLCz} q_1$ then
    $t \equiv \Ks_1 t_1 t_2 \to_s t_1$ with $t_1 \equiverased q_1$. We
    take $t' \equiv t_1$.
  \item If $q \equiv \Ss q_0 q_1 q_2 \to_{\sCLCz} q_0 q_2 (q_1 q_2)$
    then
    $t \equiv \Ss^{\vec{n}} s \la t_1,\ldots,t_k \ra \la
    \vec{r}_0,\ldots,\vec{r}_k \ra$ where the conventions regarding
    the vector notation are as in Definition~\ref{def_clc_s}, and
    $s \equiverased q_0$, and $t_i \equiverased q_1$ for
    $i=1,\ldots,k$, and $r_{i,j} \equiverased q_2$ for $i=0,\ldots,k$,
    $j=1,\ldots,i$. Thus
    \[
    t \to_s s \la \vec{r}_0 \ra \la t_1 \la
    \vec{r}_1 \ra, \ldots, t_{k} \la
    \vec{r}_k \ra \ra \equiverased q_0 q_2
    (q_1 q_2)
    \]
    and we may take~$t' \equiv s \la \vec{r}_0 \ra \la t_1
    \la \vec{r}_{1} \ra, \ldots, t_{k} \la
    \vec{r}_k \ra \ra$.
  \end{itemize}
  If~$t$ is a tuple or~$q$ is not the contracted $\CLCz$-redex, then
  the claim follows from the inductive hypothesis.
\end{proof}

The following technical lemma shows that~$\eqv_i$ may be postponed
after~$\to_s$.

\begin{lemma}\label{lem_i_postpone}
  If $t \eqv_i \cdot \to_s^* t'$ then
  $t \to_s^* \cdot \eqv_i^{\equiv} t'$.
\end{lemma}

\begin{proof}
  Suppose $t_1 \eqv_i t_2 \to_s t_3$. We proceed by induction on the
  definition of $t_2 \to_s t_3$.

  If~$t_2$ is the contracted $s$-redex then, because an $i$-redex
  ($i$-contractum) is an $i$-term, it is easy to see by inspecting
  Definition~\ref{def_clc_s} that the $i$-redex ($i$-contractum)
  in~$t_2$ must occur below a variable position of the
  $s$-redex. Since significant reduction rules are linear and their
  conditions are stable under $i$-reductions ($i$-expansions), the
  claim holds. Note that we need~$\eqv_{i}^{\equiv}$ instead
  of~$\eqv_{i}$ in the conclusion, because the $i$-redex
  ($i$-contractum) may be erased by the $s$-contraction.

  If~$t_2$ is not the $s$-redex, then $t_2 \equiv s_1 s_2$ or $t_2
  \equiv \la s_1,\ldots,s_n \ra$ with $n > 1$, and the claim is easily
  established possibly appealing to the inductive hypothesis.
\end{proof}

The next lemmas show that $i$-reductions/expansions and
$s$-reductions preserve~$\Da_{\Fs_1}$.

\begin{lemma}\label{lem_eqv_i_std}
  If $t$ is standard and $t \eqv_{i} t'$ then $t'$ is standard.
\end{lemma}

\begin{proof}
  We check that the conditions in Definition~\ref{def_standard} hold
  for every subterm~$s'$ of~$t'$. Note that because $i$-redexes and
  $i$-contracta are $i$-terms, $s'$ is an $i$-term or there is a
  subterm~$s$ of~$t$ such that $s \eqv_{i}^{\equiv} s'$.
  \begin{enumerate}
  \item If $s'$ is not an $i$-term, then there is a subterm~$s$ of~$t$
    such that $s \eqv_{i}^{\equiv} s'$. If~$s$ is an $i$-term or a
    tuple then so is~$s'$. Otherwise,~$s$ is an $s$-term
    by~\ref{std_i_or_s_or_tuple} in
    Definition~\ref{def_standard}. Then~$s'$ is also an $s$-term.
  \item Suppose $s' \equiv \Cs_1 t_0' t_1' t_2'$ with $t_0'$ in
    $s$-NF. Since~$s'$ is not an $i$-term, there is a subterm~$s$
    of~$t$ such that $s \equiv \Cs_1 t_0 t_1 t_2$ and
    $t_i \eqv_{i}^{\equiv} t_i'$ for $i=0,1,2$. Since~$t_0'$ is in
    $s$-NF and $t_0 \eqv_{i}^{\equiv} t_0'$, the term~$t_0$ is also in
    $s$-NF. Thus $t_0 \equiv \Ts_1$ or $t_0 \equiv \Fs_1$
    by~\ref{std_c_1} in Definition~\ref{def_standard}. Hence
    $t_0' \equiv \Ts_1$ or $t_0' \equiv \Fs_1$.
  \item Suppose $s' \equiv \Cs_2 t_0' t_1' t_2'$. Since~$s'$ is not an
    $i$-term, there is a subterm~$s$ of~$t$ such that
    $s \equiv \Cs_1 t_0 t_1 t_2$ and $t_i \eqv_{i}^{\equiv} t_i'$ for
    $i=0,1,2$. By~\ref{std_c_2} in Definition~\ref{def_standard} we
    have $\erase{t_1} =_\sCLC \erase{t_2}$. Hence also
    $\erase{t_1'} =_\sCLC \erase{t_2'}$, because
    $t_i \eqv_{i}^{\equiv} t_i'$ implies
    $\erase{t_i} =_\sCLC \erase{t_i'}$.
  \item Suppose $s' \equiv \Ss^{n_0,\ldots,n_k} t_0' t_1'
    t_2'$. Since~$s'$ is not an $i$-term, there is a subterm~$s$
    of~$t$ such that $s \equiv \Ss^{n_0,\ldots,n_k} t_0 t_1 t_2$ and
    $t_i \eqv_{i}^{\equiv} t_i'$ for $i=0,1,2$. By~\ref{std_s} in
    Definition~\ref{def_standard} we conclude that~$t_2$ is a tuple of
    length $n=\sum_{i=0}^kn_i$, and if $k > 1$ then~$t_1$ is a tuple
    of length~$k$. The same holds for~$t_2'$ and~$t_1'$, because a
    tuple cannot be an $i$-redex or an $i$-contractum.
  \item Suppose $s'$ is an $s$-term. There is a subterm~$s$ of~$t$
    such that $s \eqv_{i}^{\equiv} s'$. Since~$s'$ is an $s$-term, so
    is~$s$. Suppose $s' \to_s^{*} r'$. By Lemma~\ref{lem_i_postpone}
    there is~$r$ such that $s \to_s^{*} r \eqv_{i}^{\equiv}
    r'$. By~\ref{std_s_term_reduce} in Definition~\ref{def_standard},
    the term~$r$ is an $s$-term. Hence,~$r'$ is also an $s$-term.
  \item Suppose $s' \equiv \la t_1',\ldots,t_n' \ra$ with $n >
    1$. Since~$s'$ is not an $i$-term, there is a subterm~$s$ of~$t$
    such that $s \equiv \la t_1,\ldots,t_n \ra$ and
    $t_i \eqv_{i}^{\equiv} t_i'$ for
    $i=1,\ldots,n$. By~\ref{std_no_nested_tuple} in
    Definition~\ref{def_standard} none of $t_1,\ldots,t_n$ is a
    tuple. Hence, none of $t_1',\ldots,t_n'$ is a tuple either.
  \end{enumerate}
\end{proof}

\begin{lemma}\label{lem_eqv_i}
  If $t \Da_{\Fs_1}$ and $t \eqv_{i} t'$ then $t' \Da_{\Fs_1}$.
\end{lemma}

\begin{proof}
  Suppose $t' \to_s^{*} t_0'$. By Lemma~\ref{lem_i_postpone} there
  is~$t_0$ with $t \to_s^{*} t_0$ and $t_0 \eqv_{i}^\equiv
  t_0'$. Because~$t$ is strongly standard, $t_0$ is
  standard. Hence~$t_0'$ is standard by
  Lemma~\ref{lem_eqv_i_std}. Therefore~$t'$ is strongly standard.

  Suppose $t' \to_s^{*} t_0'$ with~$t_0'$ in $s$-NF. By
  Lemma~\ref{lem_i_postpone} there is~$t_0$ with
  $t \to_s^{*} t_0 \eqv_{i}^{\equiv} t_0'$. Since~$t_0'$ is in $s$-NF,
  so is~$t_0$, because an $i$-contraction or an $i$-expansion cannot
  create an $s$-redex. Since $t \Da_{\Fs_1}$ we obtain
  $t_0 \equiv \Fs_1$. Thus $t_0' \equiv t_0 \equiv \Fs_1$.
\end{proof}

\begin{lemma}\label{lem_contr_s}
  If $t \Da_{\Fs_1}$ and $t \to_s t'$ then $t' \Da_{\Fs_1}$.
\end{lemma}

\begin{corollary}\label{cor_contr}
  If $t \Da_{\Fs_1}$, $t \equiverased q$ and $q \to_{\sCLCz} q'$ then
  there is~$t'$ with $t' \equiverased q'$ and $t' \Da_{\Fs_1}$.
\end{corollary}

\begin{proof}
  Follows from Lemma~\ref{lem_erase_contr}, Lemma~\ref{lem_eqv_i} and
  Lemma~\ref{lem_contr_s}.
\end{proof}

With the above corollary we have finished the first half of the
proof. Now we need to show an analogous corollary for
$\CLCz$-expansions. First, we want to prove that $\CLCz$-expansions in
unlabeled terms may be simulated by $i$-expansions and $a$-expansions
in their strongly standard labeled variants. We have already shown in
Lemma~\ref{lem_eqv_i} that $i$-expansions preserve~$\Da_{\Fs_1}$. We
need to show that $a$-expansions also preserve~$\Da_{\Fs_1}$.

One trivial reason why $s$-expansions do not necessarily
preserve~$\Da_{\Fs_1}$ is that if $t \leftidx{{}_s}{\from} t'$
then~$t'$ may be not standard even if~$t$ is, e.g., consider
$\Fs_1 \leftidx{{}_s}{\from} \Ks_1 \Fs_1 (\Cs \Ts_1)$. A more profound
reason is that with $s$-expansion we do not sufficiently ``control''
the expansion by a rule for~$\Cs_2$.
E.g.~$\Fs_1 \leftidx{{}_s}{\from} \Cs_2 \Omega \Fs_1 (\Ks \Fs
\Omega)$. Then
$\Cs_2 \Omega \Fs_1 (\Ks \Fs \Omega) \to_s \Ks \Fs \Omega$ but
$\Ks \Fs \Omega$ does not $s$-reduce to~$\Fs_1$.

Hence, we use $a$-expansions which put additional restrictions on the
$s$-redexes, essentially implementing the labeling of expansions
described in Section~\ref{sec_overview}. They also allow to ``delay''
the reductions in a contractum of~$\Cs_2 t_0 t_1 t_1$ to facilitate
the proof of an analogon of Lemma~\ref{lem_i_postpone}.

Like in the proof of Lemma~\ref{lem_eqv_i} we show that if
$t' \to_a t$ then any reduction $t' \to_s^{*} s'$ may be simulated by
a reduction $t \to_s^{*} s$ with $s' \to_a^{\equiv} s$. The most
interesting case is when
$t' \equiv E[\Cs_2 t_0 t_1 t_1] \to_a E[t_1] \equiv t$ (where~$E$ is a
context), which is obtained from a $\CLCz$-expansion by the rule
$\Cs x y y \to y$. We now informally describe the idea for the proof
in this case. Thus suppose $t' \to_s^{*} s'$. If a contracted
$s$-redex does not overlap with a descendant\footnote{Note that
  because the rules of significant reduction are linear there may be
  at most one descendant.} of~$\Cs_2 t_0 t_1 t_1$, then the
$s$-reduction is simulated by the same $s$-reduction. If a descendant
of~$\Cs_2 t_0 t_1 t_1$ occurs inside a contracted $s$-redex, but it is
different from this redex, then the descendant must occur below a
variable position of the $s$-redex, because there are no non-root
overlaps between the rules of significant reduction. Thus we may
simulate this $s$-reduction by the same $s$-reduction. If a contracted
$s$-redex occurs inside a descendant~$\Cs_2 t_0' t_1' t_2'$
of~$\Cs_2 t_0 t_1 t_1$, but it is different from this descendant, then
it must occur in~$t_0'$,~$t_1'$ or~$t_2'$. In this case we ignore the
$s$-contraction while at all times maintaining the invariant: if
$\Cs_2 t_0' t_1' t_2'$ is a descendant of~$\Cs_2 t_0 t_1 t_1$ then
$t_1 \to_s^{*} t_1'$ and $t_1 \to_s^{*} t_2'$, and the descendant
of~$t_1$ in the simulated reduction is always identical with~$t_1$,
i.e.~$t_1$ (the $a$-contractum of $\Cs_2 t_0 t_1 t_1$) is not changed
by the simulated $s$-reduction. Finally, if a
descendant~$\Cs_2 t_0' t_1' t_2'$ of~$\Cs_2 t_0 t_1 t_1$ is
$s$-contracted, then either $\Cs_2 t_0' t_1' t_2' \to_s t_1'$ or
$\Cs_2 t_0' t_1' t_2' \to_s t_2'$. In any case we can $s$-reduce~$t_1$
to~$t_1'$ or~$t_2'$. In other words, we defer the choice of the
simulated reduction path till the descendant of the $a$-redex is
actually contracted.

\begin{definition}\label{def_a_redex}
  An $l$-term~$t'$ is an \emph{$a$-redex} and~$t$ its
  \emph{$a$-contractum}, if~$t$ is an $s$-term and one of the
  following holds:
  \begin{itemize}
  \item $t' \equiv \Cs_1 \Ts_1 t q$ and~$q$ is an $i$-term,
  \item $t' \equiv \Cs_1 \Fs_1 q t$ and~$q$ is an $i$-term,
  \item $t' \equiv \Cs_2 q t_1 t_2$, $t \to_s^{*} t_1$, $t \to_s^{*}
    t_2$ and~$q$ is an $i$-term,
  \item $t' \equiv \Ks_1 t q$ and~$q$ is an $i$-term,
  \item $t' \equiv \Ss^{\vec{n}} t_0 \la s_1,\ldots,s_{k} \ra \la
    \vec{r}_0,\ldots,\vec{r}_k \ra$ where the conventions regarding
    vector notation are as in Definition~\ref{def_clc_s}, $\erase{s_i}
    =_\sCLC \erase{s_j}$ for $i,j=1,\ldots,k$, $\erase{r_{i,j}}
    =_\sCLC \erase{r_{i',j'}}$ for $i,i'=0,\ldots,k$,
    $j=1,\ldots,n_i$, $j'=1,\ldots,n_{i'}$, none of the~$s_{i}$
    or~$r_{i,j}$ is a tuple, and $t \equiv t_0 \la \vec{r}_0 \ra \la
    s_1 \la \vec{r}_1 \ra, \ldots, s_{k} \la \vec{r}_k \ra \ra$.
  \end{itemize}
  Because of the third point, an $a$-contractum of an $a$-redex is not
  unique. The notations $\to_a$, $\to_a^{*}$, $\to_{i,a}$, etc.~are
  used accordingly. Note that any $a$-redex is an $s$-redex.
\end{definition}

\begin{lemma}\label{lem_a_to_erased_equiv}
  If $t' \to_a t$ then $t' \to_s \cdot \leftidx{{}_s^*}{\from} t$, and
  hence $\erase{t'} =_\sCLC \erase{t}$.
\end{lemma}

The above simple lemma implies that the conditions in significant
reduction rules are stable under $a$-reduction and $a$-expansion. Note
that if $t' \to_a t$ then not necessarily $t' \to_s t$ because of the
third point in Definition~\ref{def_a_redex}.

\begin{lemma}\label{lem_erase_expand}
  If~$t$ is standard, $t \equiverased q$ and
  $q \leftidx{{}_{\sCLCz}}{\from} q'$ then there is~$t'$ with
  $t' \to_{i,a}^{*} t$ and $t' \equiverased q'$.
\end{lemma}

\begin{proof}
  Induction on the size of~$t$. First assume~$t$ is not a tuple
  and~$q$ is the $\CLCz$-contractum expanded in
  $q \leftidx{{}_{\sCLCz}}{\from} q'$. If~$t$ is an $i$-term, then
  $t \equiv q \leftidx{{}_i}{\from} q'$ and we may take
  $t' \equiv q'$. If~$t$ is not an $i$-term, then it is an $s$-term
  by~\ref{std_i_or_s_or_tuple} in Definition~\ref{def_standard}. We
  have the following possibilities, depending on the rule of~$\CLCz$
  used in the expansion.
  \begin{itemize}
  \item If $q' \equiv \Cs \Ts q_1 q_2 \to_{\sCLCz} q_1 \equiv q$ then
    we take $t' \equiv \Cs_1 \Ts_1 t q_2$ and we have $t' \to_a t$ and
    $t' \equiverased q'$.
  \item If $q' \equiv \Cs \Fs q_1 q_2 \to_{\sCLCz} q_2 \equiv q$ then
    we may take $t' \equiv \Cs_1 \Fs_1 q_1 t$.
  \item If $q' \equiv \Cs q_0 q_1 q_1 \to_{\sCLCz} q_1$ then we may
    take $t' \equiv \Cs_2 q_0 t t$.
  \item If $q' \equiv \Ks q_0 q_1 \to_{\sCLCz} q_0$ then we may take
    $t' \equiv \Ks_1 t q_1$.
  \item If $q' \equiv \Ss q_0 q_1 q_2 \to_{\sCLCz} q_0 q_2 (q_1 q_2)$
    then $t \equiverased q_0 q_2 (q_1 q_2)$ and~$t$ is an
    $s$-term. Hence $t \equiv t_a t_b t_c$ with $t_a \equiverased
    q_0$, $t_b \equiverased q_2$ and $t_c \equiverased q_1 q_2$.
    Recalling the convention $\la s \ra \equiv s$ for any term~$s$, we
    may assume
    \begin{itemize}
    \item[($\star$)] $t_b \equiv \la s_1,\ldots,s_m \ra$, $t_c \equiv
      \la t_1,\ldots,t_k \ra$, for $k,m \in \Nbb_+$, if $k=1$
      then~$t_1$ is not a tuple, and if $m=1$ then~$s_1$ is not a
      tuple.
    \end{itemize}
    In other words, if e.g.~$t_b$ is a tuple, then
    $t_b \equiv \la s_1,\ldots,s_m \ra$ for some
    $s_1,\ldots,s_m$. If~$t_b$ is not a tuple then we take
    $s_1 \equiv t_b$ and consider
    $t_b \equiv \la t_b \ra \equiv \la s_1 \ra$. This is chiefly to
    reduce the number of cases to consider. Let $1 \le i \le
    k$. Because~$t_b \equiverased q_2$, we have~$s_i \equiverased q_2$
    for $i=1,\ldots,m$. Also none of $s_1,\ldots,s_m$ is a tuple, by
    condition~\ref{std_no_nested_tuple} in
    Definition~\ref{def_standard}, or by~($\star$) if $m=1$. Since
    $t_c \equiverased q_1 q_2$, we have $t_i \equiverased q_1
    q_2$. Also~$t_i$ cannot be a tuple, by
    condition~\ref{std_no_nested_tuple} in
    Definition~\ref{def_standard}, or by~($\star$) if $k=1$. Thus
    $t_i \equiv u_i \la \vec{r}_i \ra$ where $\vec{r}_i$ stands for
    $r_{i,1},\ldots,r_{i,n_i}$, and $u_i \equiverased q_1$ and
    $r_{i,j} \equiverased q_2$ for $j=1,\ldots,n_i$, where none of
    the~$r_{i,j}$ is a tuple, by definition (if $n_i = 1$) or by
    condition~\ref{std_no_nested_tuple} in
    Definition~\ref{def_standard}. By
    Lemma~\ref{lem_standard_properties} also none of $u_1,\ldots,u_k$
    is a tuple. We may thus take
    $t' \equiv \Ss^{m,n_1,\ldots,n_k} t_a \la u_1,\ldots,u_k \ra \la
    \vec{r_0},\vec{r}_1,\ldots,\vec{r}_k\ra$ where $\vec{r}_0$ stands
    for $s_1,\ldots,s_m$. We have $t' \to_a t$ and
    $t' \equiverased q'$.
  \end{itemize}
  If~$t$ is a tuple or~$q$ is not the $\CLCz$-contractum, then the
  claim follows from the inductive hypothesis.
\end{proof}

\begin{lemma}\label{lem_s_a_commute}
  If $t \leftidx{{}_a}{\from} \cdot \to_s t'$ and $t$ is standard then
  $t \reduces_s \cdot \leftidx{{}_a^\equiv}{\from} t'$.
\end{lemma}

\begin{proof}
  Suppose $t' \to_s t_1'$, $t' \to_a t$ and~$t$ is standard. By
  induction on the definition of $t' \to_s t_1'$ we show that there
  is~$t_1$ with $t \reduces_s t_1$ and $t_1' \to_a^\equiv t_1$. The
  base case is when the $s$-contraction in $t' \to_s t_1'$ occurs at
  the root.

  If the $s$-contraction occurs at the root, but the $a$-contraction
  in $t' \to_a t$ does not occur at the root, then it is easy to see
  by inspecting the definitions that the $a$-redex in~$t_1'$ must
  occur below a variable position of the $s$-redex. Since significant
  reduction rules are linear and their conditions are stable under
  $a$-reduction, the claim holds in this case.

  Assume that both the $s$-contraction and the $a$-contraction occur
  at the root. If $t' \equiv \Cs_2 q s_1 s_2 \to_a t$ then
  $t \to_s^{*} s_1$, $t \to_s^{*} s_2$ and the $s$-contraction of~$t'$
  yields either~$s_1$ or~$s_2$. We may thus take either
  $t_1 \equiv s_1$ or $t_1 \equiv s_2$, and we have
  $t \to_s^{*} t_1 \equiv t_1'$. If
  $t' \equiv \Cs_1 \Ts_1 t q \to_a t$ then the $s$-contraction must be
  by the first rule of~$\CLCs$, so $t_1' \equiv t$ and we may take
  $t_1 \equiv t_1' \equiv t$. All other cases are analogous.

  If neither the $s$-contraction nor the $a$-contraction occurs at the
  root, then the claim is easily established, possibly appealing to
  the inductive hypothesis.

  Finally, assume that the $a$-contraction occurs at the root, but the
  $s$-contraction does not occur at the root. We have the following
  possibilities.
  \begin{itemize}
  \item If $t' \equiv \Cs_1 \Ts_1 t q \to_a t$ then the
    $s$-contraction must occur inside~$t$. So~$t \to_s t_1$ for
    some term~$t_1$. Note that~$t$ is an $s$-term by definition of
    $a$-contraction. Therefore~$t_1$ is also an $s$-term,
    by~\ref{std_s_term_reduce} in
    Definition~\ref{def_standard}. Thus~$t_1$ satisfies the required
    conditions.
  \item If $t' \equiv \Cs_2 q s_1 s_2 \to_a t$ then $t \to_s^{*} s_1$,
    $t \to_s^{*} s_2$ and the $s$-contraction must occur inside~$s_1$
    or~$s_2$. We may take $t_1 \equiv t$ and we still have
    $t_1' \to_a t_1$.
  \item The cases $t' \equiv \Cs_1 \Fs_1 q t \to_a t$ and
    $t' \equiv K_1 t q \to_a t$ are analogous to the first case.
  \item If
    $t' \equiv \Ss^{\vec{n}} t_0 \la s_1,\ldots,s_{k} \ra \la
    \vec{r}_0,\ldots,\vec{r}_k \ra$ then
    $\erase{s_i} =_\sCLC \erase{s_j}$,
    $\erase{r_{i,j}} =_\sCLC \erase{r_{i',j'}}$ for $i,j,i',j'$ as in
    Definition~\ref{def_a_redex}, none of the $s_i$ or $r_{i,j}$ is a
    tuple, and
    $t \equiv t_0 \la \vec{r}_0 \ra \la s_1 \la \vec{r}_1 \ra, \ldots,
    s_{k} \la \vec{r}_k \ra \ra$. The $s$-contraction must occur
    inside one of the~$s_i$ or the~$r_{i,j}$, or in~$t_0$. For
    instance, assume $s_1 \to_s s_1'$. Since~$s_1$ is a subterm
    of~$t$ and it is not a tuple, it cannot $s$-reduce to a tuple by
    Definition~\ref{def_standard}. Hence~$s_1'$ is not a tuple. Take
    $t_1 \equiv t_0 \la \vec{r}_0 \ra \la s_1' \la \vec{r}_1 \ra, s_2
    \la \vec{r}_2 \ra, \ldots, s_{k} \la \vec{r}_k \ra \ra$. Note that
    $t \to_s t_1$. Thus~$t_1$ is an $s$-term, because~$t$ is an
    $s$-term and it $s$-reduces only to $s$-terms,
    by~\ref{std_s_term_reduce} in Definition~\ref{def_standard}.
  \end{itemize}
\end{proof}

\begin{corollary}\label{cor_s_a_commute}
  If $t \leftidx{{}_a}{\from} \cdot \reduces_s t'$ and~$t$ is strongly
  standard then $t \reduces_s \cdot \leftidx{{}_a^\equiv}{\from} t'$.
\end{corollary}

\begin{lemma}\label{lem_properly_in_a_redex}
  If $r$ is a strongly standard $a$-contractum of an $a$-redex~$r'$,
  and~$s'$ is a proper subterm of~$r'$, then~$s'$ is standard.
\end{lemma}

\begin{proof}
  It suffices to show that~$s'$ is a subterm of some standard term.
  \begin{itemize}
  \item Suppose $r' \equiv \Cs_1 \Ts_1 r q \to_a r$ with~$q$ an
    $i$-term. Both~$\Cs_1 \Ts_1 r$ and~$q$ are standard and~$s'$ is a
    subterm of one of them. The cases $r' \equiv \Cs_1 \Fs_1 q r$ and
    $r' \equiv \Ks_1 r q$ are analogous.
  \item Suppose $r' \equiv \Cs_2 q r_1 r_2 \to_a r$ with~$q$ an
    $i$-term. Because $r \to_s^{*} r_i$ and~$r$ is strongly standard,
    $r_1,r_2$ are standard. Also~$q$ is an $i$-term. This implies that
    $\Cs_2 q r_1$ is also standard. Since~$s'$ occurs in~$\Cs_2 q r_1$ or
    $r_2$, it is standard.
  \item Suppose
    $r' \equiv \Ss^{\vec{n}} t_0 \la s_1,\ldots,s_k \ra \la
    \vec{r}_0,\ldots,\vec{r}_k \ra \to_a t_0 \la \vec{r}_0 \ra \la
    s_1 \la \vec{r}_1 \ra, \ldots, s_k \la \vec{r}_k \ra \ra$. The
    term~$t_0$ and each of~$s_i$ and~$r_{i,j}$ (with $i,j$ as in
    Definition~\ref{def_a_redex}) is standard. Note that none of~$s_i$
    or~$r_{i,j}$ is a tuple by Definition~\ref{def_a_redex}. Since
    each~$s_i$ is also standard, by inspecting
    Definition~\ref{def_standard} we may conclude
    that~$\la s_1,\ldots,s_k\ra$ is standard. Similarly
    $\la \vec{r}_0,\ldots,\vec{r}_k \ra$ is standard. Also
    $\Ss^{\vec{n}} t_0 \la s_1,\ldots,s_k \ra$ is standard. This
    implies that~$s'$ is standard, because it occurs in
    $\Ss^{\vec{n}} t_0 \la s_1,\ldots,s_k \ra$ or
    $\la \vec{r}_0,\ldots,\vec{r}_k \ra$.
  \end{itemize}
\end{proof}

\begin{lemma}\label{lem_a_preserves_s_term}
  If $s$ is an $s$-term and $s' \to_a s$ then $s'$ is also an
  $s$-term.
\end{lemma}

\begin{proof}
  Induction on the structure of~$s$.
\end{proof}

\begin{lemma}\label{lem_expand_a_std}
  If~$t$ is strongly standard and $t' \to_a t$ then~$t'$ is standard.
\end{lemma}

\begin{proof}
  We check that the conditions in Definition~\ref{def_standard} hold
  for every subterm~$s'$ of~$t'$. We may assume that~$s'$ does not
  occur in~$t$, as otherwise the claim follows from the fact that~$t$
  is standard. Therefore,~$s'$ occurs in the $a$-redex contracted in
  $t' \to_a t$, or the $a$-redex occurs inside~$s'$. If $s'$ is a
  proper subterm of the $a$-redex, then our claim holds by
  Lemma~\ref{lem_properly_in_a_redex}. Hence, we may assume that the
  $a$-redex~$r'$ is a subterm of~$s'$. Then $s' \to_a s$ with~$s$ a
  subterm of~$t$ (so~$s$ is strongly standard).
  \begin{enumerate}
  \item Suppose~$r$ is the $a$-contractum of~$r'$ and $s' \to_a
    s$. By Definition~\ref{def_a_redex}, the term~$r$ is an
    $s$-term. Thus~$s$ cannot be an $i$-term. If~$s$ is a tuple, then
    so is~$s'$. Otherwise,~$s$ is an $s$-term,
    by~\ref{std_i_or_s_or_tuple} in
    Definition~\ref{def_standard}. Hence $s'$ is also an $s$-term by
    Lemma~\ref{lem_a_preserves_s_term}.
  \item Suppose $s' \equiv \Cs_1 t_0' t_1' t_2'$ and~$t_0'$ is in
    $s$-NF. If $s' \equiv r'$ then $s' \equiv \Cs_1 \Ts_1 t_1' t_2'$
    or $s' \equiv \Cs_1 \Fs_1 t_1' t_2'$, hence $t_0' \equiv \Ts_1$ or
    $t_0' \equiv \Fs_1$. If~$r'$ is a proper subterm of~$s'$,
    then~$r'$ must be a subterm of~$t_1'$ or~$t_2'$, because
    $a$-redexes are not in $s$-NF. Thus,
    $s' \to_a s \equiv \Cs_1 t_0' t_1 t_2$ for some terms $t_1,t_2$,
    where~$s$ is a subterm of~$t$. Hence, $t_0' \equiv \Ts_1$ or
    $t_0' \equiv \Fs_1$ by~\ref{std_c_1} in
    Definition~\ref{def_standard}.
  \item Suppose $s' \equiv \Cs_2 t_0' t_1' t_2'$. If $s' \equiv r'$
    then $s \to_s^{*} t_1'$ and $s \to_s^{*} t_2'$. Hence
    $\erase{t_1'} =_\sCLC \erase{s} =_\sCLC \erase{t_2'}$. If
    $s' \not\equiv r'$ then $s \equiv \Cs_2 t_0 t_1 t_2$ with
    $t_i' \to_a^{\equiv} t_i$. Because~$s$ is standard,
    $\erase{t_1} =_\sCLC \erase{t_2}$ by~\ref{std_c_2} in
    Definition~\ref{def_standard}. Thus also
    $\erase{t_1'} =_\sCLC \erase{t_2'}$ by
    Lemma~\ref{lem_a_to_erased_equiv}.
  \item Suppose $s' \equiv \Ss^{n_0,\ldots,n_k} t_0' t_1' t_2'$. If
    $s' \equiv r'$, then
    $s' \equiv \Ss^{n_0,\ldots,n_k} t_0' \la s_1,\ldots,s_{k} \ra \la
    \vec{r}_0,\ldots,\vec{r}_k\ra$, as in
    Definition~\ref{def_a_redex}, so the claim holds. If~$r'$ is a
    proper subterm of~$s'$, then
    $s' \equiv \Ss^{n_0,\ldots,n_k} t_0' t_1' t_2' \to_a s \equiv
    \Ss^{n_0,\ldots,n_k} t_0 t_1 t_2$ where $t_i' \to_a^{\equiv} t_i$
    for $i=0,1,2$, and~$s$ is a subterm of~$t$. By~\ref{std_s} in
    Definition~\ref{def_standard}, the term~$t_2$ is a tuple of
    length~$\sum_{i=0}^kn_i$, and if $k > 1$ then~$t_1$ is a tuple of
    length~$k$. Since an $a$-contractum is an $s$-term, and hence not
    a tuple, $t_2$ is not an $a$-contractum, and if $k > 1$ then~$t_1$
    is not an $a$-contractum. Thus we may conclude that~$t_2'$ is a
    tuple of length~$\sum_{i=0}^kn_i$, and if $k > 1$ then~$t_1'$ is a
    tuple of length~$k$.
  \item Suppose $s'$ is an $s$-term and $s' \to_s^{*} s_1'$. Because
    also $s' \to_a s$ and $s$ is strongly standard, by
    Corollary~\ref{cor_s_a_commute} there is~$s_1$ with
    $s_1' \to_a^{\equiv} s_1$ and $s \to_s^{*} s_1$. By
    Definition~\ref{def_a_redex}, the term~$s$ is an $s$-term,
    so~$s_1$ is also an $s$-term by~\ref{std_s_term_reduce} in
    Definition~\ref{def_standard}. So~$s_1'$ is an $s$-term by
    Lemma~\ref{lem_a_preserves_s_term}.
  \item Suppose $s' \equiv \la t_1',\ldots,t_n' \ra$ with $n > 1$. We
    have $s' \to_a s \equiv \la t_1,\ldots,t_n \ra$ where $t_i'
    \to_a^{\equiv} t_i$ for
    $i=1,\ldots,n$. By~\ref{std_no_nested_tuple} in
    Definition~\ref{def_standard}, none of $t_1,\ldots,t_n$ is a
    tuple. Thus it is easy to see by inspecting
    Definition~\ref{def_a_redex} that none of~$t_1',\ldots,t_n'$ can
    be a tuple.
  \end{enumerate}
\end{proof}

\begin{lemma}\label{lem_expand_a}
  If $t \Da_{\Fs_1}$ and $t' \to_a t$ then $t' \Da_{\Fs_1}$.
\end{lemma}

\begin{proof}
  Suppose $t' \to_s^{*} t_0'$. By Corollary~\ref{cor_s_a_commute}
  there is~$t_0$ with $t \to_s^{*} t_0$ and $t_0' \to_a^{\equiv}
  t_0$. Since~$t$ is strongly standard, so is~$t_0$. Therefore,~$t_0'$
  is standard by Lemma~\ref{lem_expand_a_std}.

  Suppose $t' \to_s^{*} t_0'$ with~$t_0'$ in $s$-NF. By
  Corollary~\ref{cor_s_a_commute} there is~$t_0$ with
  $t \to_s^{*} t_0$ and $t_0' \to_a^{\equiv} t_0$. Since an $a$-redex
  is an $s$-redex and~$t_0'$ is in $s$-NF, we conclude that
  $t_0' \equiv t_0$. But then $t_0' \equiv t_0 \equiv \Fs_1$, because
  $t \Da_{\Fs_1}$.
\end{proof}

\begin{corollary}\label{cor_expand}
  If $t \Da_{\Fs_1}$, $t \equiverased q$ and
  $q \leftidx{{}_{\sCLCz}}{\from} q'$ then there is~$t'$ with
  $t' \Da_{\Fs_1}$ and $t' \equiverased q'$.
\end{corollary}

\begin{proof}
  Follows from Lemma~\ref{lem_erase_expand}, Lemma~\ref{lem_eqv_i} and
  Lemma~\ref{lem_expand_a}.
\end{proof}

\begin{lemma}\label{lem_good_reduction}
  If~$t$ has no $s$-NFs other than~$\Fs_1$ then
  $t \to_{s{-}}^{*} \Fs_1$.
\end{lemma}

\begin{proof}
  Since $s$-reduction is terminating, by reducing $s$-redexes outside
  any tuples and not using the second rule for~$\Cs_2$ we will
  ultimately obtain a term~$t'$ with all $s$-redexes inside tuples,
  and such that $t \to_{s{-}}^{*} t'$. Note that an $s$-redex in~$t'$
  may only occur inside a tuple, because any $s$-redex by the second
  rule for~$\Cs_2$ is also an $s{-}$-redex by the first rule
  for~$\Cs_2$. If~$t'$ is in $s$-NF then $t' \equiv \Fs_1$. Otherwise,
  any $s$-NF of~$t'$ must contain a tuple, because $s$-reduction
  inside a tuple cannot erase this tuple or create an $s$-redex
  outside of it. But since any $s$-NF of~$t'$ is an $s$-NF of~$t$,
  this contradicts the fact that~$t$ has no $s$-NFs other
  than~$\Fs_1$.
\end{proof}

We now have everything we need to show the central lemma of the
confluence proof.

\begin{lemma}\label{lem_f_nf}
  The system~$\sCLC$ is $\Fs$-normal, i.e., if $q =_\sCLC \Fs$ then
  $q \reduces_\sCLC \Fs$.
\end{lemma}

\begin{proof}
  If $q =_\sCLC \Fs$ then by Lemma~\ref{lem_clc_equivalent} we have
  $q =_{\sCLCz} \Fs$. Note that $\Fs_1 \Da_{\Fs_1}$ and
  $\erase{\Fs_1} \equiv \Fs$. Thus, using Corollary~\ref{cor_contr}
  and Corollary~\ref{cor_expand} it is easy to show by induction on
  the length of $q =_{\sCLCz} \Fs$ that there is~$t$ with
  $t \equiverased q$ and $t \Da_{\Fs_1}$. By
  Lemma~\ref{lem_good_reduction} we have $t \to_{s{-}}^{*} \Fs_1$. But
  then, because $t \equiverased q$, using
  Lemma~\ref{lem_s_minus_erase} it is easy to show by induction on the
  length of $t \to_{s{-}}^{*} \Fs_1$ that
  $q \equiv \erase{t} \reduces_\sCLC \erase{\Fs_1} \equiv \Fs$.
\end{proof}

It remains to derive the confluence of~$\CLC$ and~$\CLCp$ from
Lemma~\ref{lem_f_nf}. We use a trick with an auxiliary term rewriting
system~$R$, in a way similar to how the confluence of~$\CLCp$ is
derived from the condition~$\Ts \ne_{\sCLCp} \Fs$
in~\cite{Vrijer1999}. The idea is to eliminate the non-trivial overlap
between the rules of~$\CLC$ by imposing additional side conditions.

\begin{definition}
  The term rewriting system~$R$ is defined by the following rules:
  \[
  \begin{array}{rclcrcl}
    \Ks x y &\to& x &\quad\quad& \Cs \Ts x y &\to& x \\
    \Ss x y z &\to& x z (y z) &\quad\quad& \Cs z x y &\to& y \quad\From\quad z =_{\sCLC} \Fs \\
    &&&& \Cs z x y &\to& x \quad\From\quad z \ne_{\sCLC} \Fs \land x =_{\sCLC} y
  \end{array}
  \]
\end{definition}

\begin{lemma}\label{lem_clc_0_to_r}
  If $q \contr_{\sCLCz} q'$ then $q \contr_R q'$.
\end{lemma}

\begin{lemma}\label{lem_r_to_clc}
  If $q \contr_R q'$ then $q \reduces_\sCLC q'$.
\end{lemma}

\begin{proof}
  Follows from definitions and Lemma~\ref{lem_f_nf}.
\end{proof}

\begin{lemma}\label{lem_r_confluent}
  The system~$R$ is confluent.
\end{lemma}

\begin{proof}
  Because $\Ts \ne_\sCLC \Fs$ by Lemma~\ref{lem_f_nf}, the system~$R$
  is weakly orthogonal (i.e.~it is left-linear and all its critical
  pairs are trivial). By Lemma~\ref{lem_r_to_clc} the conditions are
  stable under reduction. Weakly orthogonal conditional term rewriting
  systems whose conditions are stable under reduction are
  confluent~\cite[Chapter 4]{Terese2003}.
\end{proof}

\begin{theorem}
  The systems~$\CLC$ and~$\CLCp$ are confluent.
\end{theorem}

\begin{proof}
  Since $q_1 \contr_{\sCLC} q_2$ implies $q_1 \contr_{\sCLCp} q_2$, it
  suffices to show that if $q_1 =_{\sCLCp} q_2$ then there is~$q$ with
  $q_1 \reduces_\sCLC q$ and $q_2 \reduces_\sCLC q$. So suppose $q_1
  =_{\sCLCp} q_2$. Then by Lemma~\ref{lem_clc_equivalent} we have $q_1
  =_{\sCLCz} q_2$. By Lemma~\ref{lem_clc_0_to_r} we obtain $q_1 =_R
  q_2$. By Lemma~\ref{lem_r_confluent} there is~$q$ with $q_1
  \reduces_R q$ and $q_2 \reduces_R q$. By Lemma~\ref{lem_r_to_clc} we
  have $q_1 \reduces_\sCLC q$ and $q_2 \reduces_\sCLC q$.
\end{proof}

\bibliography{biblio}{}

\end{document}